\documentclass[aps,prx,twocolumn,showpacs,amsmath,notitlepage,amssymb,superscriptaddress]{revtex4-1}

\usepackage{bm}        
\usepackage{amssymb}   
\usepackage{amsmath}
\usepackage{amsthm}
\usepackage{graphicx}  
\usepackage{color}

\definecolor{dartmouthgreen}{rgb}{0.05, 0.5, 0.06}

\newcommand{\bra}[1]{\langle {#1} |}
\newcommand{\ket}[1]{| {#1} \rangle}
\newcommand{\braket}[2]{\langle {#1} |{#2} \rangle}

\newtheorem{lemma}{Lemma}
\newtheorem{definition}{Definition}
\newtheorem{theorem}{Theorem}
\newtheorem{proposition}{Proposition}
\newtheorem{corollary}{Corollary}

\begin{document}

\title{Perfect discrimination of non-orthogonal quantum states with posterior classical partial information}
\author{Seiseki Akibue}
\email{seiseki.akibue.rb@hco.ntt.co.jp}
 \affiliation{NTT Communication Science Laboratories, NTT Corporation 3-1 Morinosato Wakamiya, Atsugi-shi, Kanagawa 243-0124, JAPAN}
 
\author{Go Kato}
\email{go.kato.gm@hco.ntt.co.jp}
 \affiliation{NTT Communication Science Laboratories, NTT Corporation 3-1 Morinosato Wakamiya, Atsugi-shi, Kanagawa 243-0124, JAPAN}

\author{Naoki Marumo}
\email{naoki.marumo.ec@hco.ntt.co.jp}
\affiliation{NTT Communication Science Laboratories, NTT Corporation 3-1 Morinosato Wakamiya, Atsugi-shi, Kanagawa 243-0124, JAPAN}

\date{\today}

\begin{abstract}
The indistinguishability of non-orthogonal pure states lies at the heart of quantum information processing. Although the indistinguishability reflects the impossibility of measuring complementary physical quantities by a single measurement, we demonstrate that the distinguishability can be perfectly retrieved simply with the help of posterior classical partial information. We demonstrate this by showing an ensemble of non-orthogonal pure states such that a state randomly sampled from the ensemble can be perfectly identified by a single measurement with help of the post-processing of the measurement outcomes and additional partial information about the sampled state, i.e., the label of subensemble from which the state is sampled. When an ensemble consists of two subensembles, we show that the perfect distinguishability of the ensemble with the help of the post-processing can be restated as a matrix-decomposition problem. Furthermore, we give the analytical solution for the problem when both subensembles consist of two states.
\end{abstract}

\maketitle
\section{Introduction}
The existence of non-orthogonal pure states is a peculiar feature of quantum mechanics. Indeed, an ensemble of them is neither perfectly cloned \cite{nocloning1, nocloning2} nor perfectly distinguishable \cite{minerror_discrimination, unambiguous_discrimination1, unambiguous_discrimination2, unambiguous_discrimination3, maxconfident_discrimination}. This is in contrast to classical theories, which assume that any ensemble of distinct pure states, each of which is not a probabilistic mixture of different states, is perfectly distinguishable in principle. While the non-orthogonality of pure states has its origin purely in quantum mechanics, we investigate its classical aspect in this paper.

From a practical point of view, the indistinguishability of non-orthogonal pure states restricts our ability to transmit information \cite{Holevo}; conversely, it enables extremely secure designs of banknotes \cite{Qmoney} and secret key distribution \cite{BB84}. For example, in the quantum key distribution (QKD) protocol proposed in \cite{BB84}, a secret bit is encoded in a basis randomly chosen from two complementary bases, $\mathbb{S}^{(A)}=(\ket{0},\ket{1})$ and $\mathbb{S}^{(B)}=(\ket{+},\ket{-})$, where $\ket{\pm}=\frac{1}{\sqrt{2}}(\ket{0}\pm\ket{1})$. An eavesdropper cannot intercept the secret bit perfectly if she does not know which basis is used since a state in $\mathbb{S}^{(A)}$ and that in $\mathbb{S}^{(B)}$ are non-orthogonal. Moreover, even if she is informed of the label of the chosen basis, $X\in\{A,B\}$, after the quantum state encoding the secret bit is destroyed by her measurement, she cannot intercept the secret bit perfectly owing to the complementarity of measurement: accurate measurement of one physical quantity entails inaccurate measurement of another complementary quantity (see Fig.~\ref{fig:setting}). Thus, it seems that a state randomly sampled from non-orthogonal pure states cannot be identified perfectly even if classical partial information about the sampled state is available after measurement of the state is performed.

\begin{figure}
 \centering
  \includegraphics[height=.18\textheight]{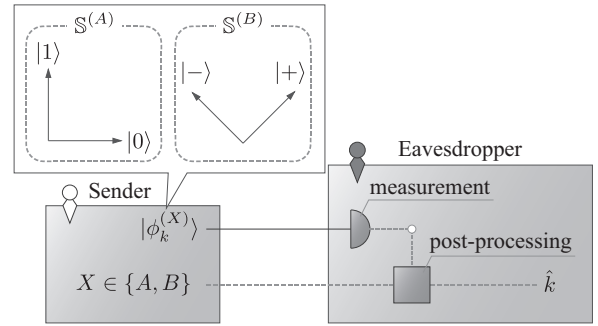}
  \caption{Indistinguishability of non-orthogonal pure states in a QKD-like protocol. First, the sender randomly chooses label $X\in\{A,B\}$ and encodes his secret bit in a basis state of $\mathbb{S}^{(X)}$. Second, the eavesdropper intercepts the state transmitted from the sender and measures it. She cannot identify the transmitted state perfectly even if she can process her measurement outcomes with label $X$.}
\label{fig:setting}
\end{figure} 

Contrary to such an intuition, in this paper, we show that such classical partial information is sometimes sufficient for accomplishing perfect discrimination of non-orthogonal pure states. Suppose that a state is randomly sampled from an ensemble of pure states, $\mathbb{S}$, consisting of two a priori known subensembles $\mathbb{S}^{(A)}$ and $\mathbb{S}^{(B)}$. 
First, we give an example of a pair of subensembles, $(\mathbb{S}^{(A)},\mathbb{S}^{(B)})$, such that $\mathbb{S}$ is an ensemble of non-orthogonal pure states but the sampled state can be perfectly identified by the classical post-processing of the measurement outcomes with the label of the subensemble, $X\in\{A,B\}$, from which the state is sampled. Second, we investigate a standard pair, $(\mathbb{S}^{(A)},\mathbb{S}^{(B)})$, which is trivially distinguishable by the post-processing. Third, we give necessary conditions for $(\mathbb{S}^{(A)},\mathbb{S}^{(B)})$ to be perfectly distinguishable by the post-processing. The conditions imply that the first example we gave can be considered as a maximally non-orthogonal distinguishable pair in the smallest Hilbert space. Finally, we show that the perfect distinguishability with the help of the post-processing can be restated as a matrix-decomposition problem, and also give the analytical solution for the problem when $|\mathbb{S}^{(A)}|=|\mathbb{S}^{(B)}|=2$. The result also implies that every perfectly distinguishable pair with the help of post-processing can be embedded in a larger Hilbert space as a standard pair.

Note that the state discrimination with the help of the post-processing has been investigated in \cite{post-processing1,post-processing2,post-processing3}, motivated by the analysis of quantum cryptographic protocols. In \cite{post-processing1} and \cite{post-processing2}, the optimal discrimination of basis states (or their probabilistic mixtures) was investigated, where the perfect discrimination is impossible in general. In \cite{post-processing3}, further investigations concerning the optimal measurement for the imperfect state discrimination were done. In contrast, we focus on the perfect discrimination of general pure states in this paper.

\section{Definitions}
We consider a quantum system represented by finite dimensional Hilbert space $\mathcal{H}$. The two a priori known ensembles of distinguishable pure states are described by indexed sets of orthonormal vectors, $\mathbb{S}^{(X)}=\big(\ket{\phi_{k}^{(X)}}\in\mathcal{H}\big)_{k\in \mathbb{K}^{(X)}}$ ($X\in\{A,B\}$), where $\mathbb{K}^{(X)}=\{0,1,\dots,|\mathbb{S}^{(X)}|-1\}$ 
 for $X\in\{A,B\}$. We suppose that the state of $\mathcal{H}$ is randomly sampled from ensemble $\mathbb{S}$ consisting of $\mathbb{S}^{(A)}$ and $\mathbb{S}^{(B)}$.

Measurement performed on $\mathcal{H}$ is described by a positive operator valued measure (POVM) over a finite set $\Omega$ \cite{minerror_discrimination}, $\Big(M_\omega\in P(\mathcal{H})\Big)_{\omega\in\Omega}$, such that $\sum_{\omega\in\Omega}M_\omega=I$, where $P(\mathcal{H})$ and $I$ represent the set of positive semi-definite operators and the identity operator on $\mathcal{H}$, respectively. After the measurement, the label of the subensemble, $X\in\{A,B\}$, from which the state is sampled is recieved, and one processes measurement outcome $\omega$ and $X$ to guess $k$ as $\hat{k}=f^{(X)}(\omega)$, where $f^{(X)}:\Omega\rightarrow\mathbb{K}^{(X)}$ for $X\in\{A,B\}$. 

Thus, pair $(\mathbb{S}^{(A)},\mathbb{S}^{(B)})$ is perfectly distinguishable by the post-processing if and only if there exist POVM $\big(M_\omega\big)_{\omega\in\Omega}$ and post-processing $\big(f^{(X)}\big)_{X\in \{A,B\}}$ such that
\begin{equation}
\label{eq:defofperfect1}
 \forall X\in\{A,B\},\forall k\in\mathbb{K}^{(X)}, \sum_{\omega\in f^{(X)-1}(k)}\bra{\phi^{(X)}_k}M_\omega\ket{\phi^{(X)}_k}=1.
\end{equation}
Note that a more general post-processing including probabilistic processing does not change the condition for the perfect distinguishability as shown in Appendix \ref{appendix:probpostprocessing}.

\section{Measurement table}
If $(\mathbb{S}^{(A)},\mathbb{S}^{(B)})$ is perfectly distinguishable by the post-processing, we can construct a measurement table representing the POVM and the classical post-processing. The measurement table is POVM over $\mathbb{K}:=\mathbb{K}^{(A)}\times\mathbb{K}^{(B)}$, $\big(M_{ab}\big)_{(a,b)\in\mathbb{K}}$ such that
\begin{equation}
 M_{ab}=\sum_{\omega\in\mathbf{f}^{-1}((a,b))}M_\omega,
\end{equation}
where $\mathbf{f}(\omega)=(f^{(A)}(\omega),f^{(B)}(\omega))$.
We can verify that $\big(M_{ab}\big)_{(a,b)\in\mathbb{K}}$ is a valid POVM, i.e., it is an indexed set of positive semi-definite operators and the sum of the elements is the identity operator. Eq.~\eqref{eq:defofperfect1} implies that
\begin{eqnarray}
\big(\forall a\in\mathbb{K}^{(A)}, &\sum_{b\in\mathbb{K}^{(B)}}\bra{\phi^{(A)}_{a}}M_{ab}\ket{\phi^{(A)}_{a}}=1\big)\nonumber\\
\label{eq:mtablecond1}
\wedge\big(\forall b\in\mathbb{K}^{(B)}, &\sum_{a\in\mathbb{K}^{(A)}}\bra{\phi^{(B)}_{b}}M_{ab}\ket{\phi^{(B)}_{b}}=1\big),
\end{eqnarray}
or equivalently,
\begin{eqnarray}
\big(\forall\{a,a'|a\neq a'\}\subseteq\mathbb{K}^{(A)},\forall b\in\mathbb{K}^{(B)}, \ket{\phi^{(A)}_{a'}}\in\ker(M_{ab})\big)\nonumber\\
\wedge\big(\forall\{b,b'|b\neq b'\}\subseteq\mathbb{K}^{(B)},\forall a\in\mathbb{K}^{(A)}, \ket{\phi^{(B)}_{b'}}\in\ker(M_{ab})\big).\nonumber\\
\label{eq:mtablecond2}
\end{eqnarray}
Conversely, if there exists a measurement table satisfying Eq.~(\ref{eq:mtablecond1}) or (\ref{eq:mtablecond2}) for $(\mathbb{S}^{(A)},\mathbb{S}^{(B)})$, it is perfectly distinguishable by the post-processing. We give an example of a measurement table which perfectly distinguishes an ensemble of non-orthogonal pure states in Table~\ref{table:ex1}, where we use notation $[\ket{\psi}]=[\psi]:=\ket{\psi}\bra{\psi}$.

\begin{table}
\begin{center}
\renewcommand
\arraystretch{1.5}
\begin{tabular}{cp{35mm}p{35mm}p{10mm}}\hline\hline
 & $\ket{0+2}$ & $\ket{0-2}$  \\ \hline
$\ket{0+1}$ & $M_{00}=\left[\frac{\sqrt{3}}{2}\ket{0+1+2}\right]$ & $M_{01}=\left[\frac{\sqrt{3}}{2}\ket{0+1-2}\right]$  \\ 
$\ket{0-1}$ & $M_{10}=\left[\frac{\sqrt{3}}{2}\ket{0-1+2}\right]$ & $M_{11}=\left[\frac{\sqrt{3}}{2}\ket{0-1-2}\right]$  \\ 
\hline\hline
\end{tabular}
\vspace{0.5cm}
\caption{Measurement table to distinguish $\mathbb{S}^{(A)}=(\ket{0+1},\ket{0-1})$ and $\mathbb{S}^{(B)}=(\ket{0+2},\ket{0-2})$, where $\ket{0+1+2}$ represents normalized state $\frac{1}{\sqrt{3}}(\ket{0}+\ket{1}+\ket{2})$. We can easily check that $\big(M_{ab}\big)$ is a valid POVM and satisfies Eq.~\eqref{eq:mtablecond2}.
}
\label{table:ex1} 
\end{center}
\end{table}


\section{Standard pair}
We define a standard pair, $(\mathbb{S}^{(A)},\mathbb{S}^{(B)})=\big(\big(\ket{\Phi^{(A)}_a}\big)_{a\in\mathbb{K}^{(A)}},\big(\ket{\Phi^{(B)}_b}\big)_{b\in\mathbb{K}^{(B)}}\big)$, which is trivially distinguishable by the post-processing as follows.
\begin{definition}
For $\mathbb{S}\subseteq\mathcal{Y}$, where $\mathcal{Y}=\mathbb{C}^{|\mathbb{K}|}$ is a Hilbert space spanned by orthonormal basis $\{\ket{ab}\}_{(a,b)\in\mathbb{K}}$, $(\mathbb{S}^{(A)},\mathbb{S}^{(B)})$ is called a standard pair if their elements are represented by
\begin{eqnarray}
 \label{eq:trivialpair}
 \ket{\Phi^{(A)}_a}=\sum_b \alpha_{ab}\ket{ab}\ \ \wedge\ \ 
 \ket{\Phi^{(B)}_b}=\sum_a \beta_{ab}\ket{ab},
\end{eqnarray}
where $\sum_b|\alpha_{ab}|^2=1$ and $\sum_a|\beta_{ab}|^2=1$.
\end{definition}
We can easily verify that the standard pair is perfectly distinguishable by measurement table $\big(M_{ab}=[ab]\big)$. In addition to the standard pair, we can verify that if $(\mathbb{S}^{(A)},\mathbb{S}^{(B)})$ can be embedded in a larger Hilbert space as a standard pair, it is also perfectly distinguishable by the post-processing as stated in the following proposition.

\begin{proposition}
\label{prop:distinguishableset1}
 Let the reduced Hilbert space of $\mathcal{H}$ be $\mathcal{X}:={\rm span}(\mathbb{S})$. If there exists isometry $V:\mathcal{X}\rightarrow\mathcal{Y}$ such that $\big(\big(V\ket{\phi^{(A)}_a}\big),\big(V\ket{\phi^{(B)}_b}\big)\big)$ is a standard pair, $(\mathbb{S}^{(A)},\mathbb{S}^{(B)})$ is perfectly distinguishable by the post-processing.
\end{proposition}
\begin{proof}
 By a straightforward calculation, we can verify that the following measurement table distinguishes $(\mathbb{S}^{(A)},\mathbb{S}^{(B)})$ perfectly:
\begin{equation}
 M_{ab}=
 \left\{ \begin{array}{ll}
    P_{\bot}+V^{\dag}[00]V& (a=b=0) \\
    V^{\dag}[ab]V& ({\rm otherwise}).
  \end{array} \right.
\end{equation}
where $P_{\bot}$ is the hermitian projection to the orthogonal complement of $\mathcal{X}$.
\end{proof}

Note that if $(\mathbb{S}^{(A)},\mathbb{S}^{(B)})$ is perfectly distinguishable by measurement table $\big(M_{ab}\big)$ consisting of rank-$r$ operators with $r\leq1$, it can always be embedded in a larger Hilbert space as a standard pair by using Naimark's extension as follows:
Let $M_{ab}=[\tilde{\psi}_{ab}]$, where $\ket{\tilde{\psi}_{ab}}\in\mathcal{H}$ is an unnormalized state. Define isometry $V=\sum_{(a,b)\in\mathbb{K}}\ket{ab}\bra{\tilde{\psi}_{ab}}$. Then $\big(\big(V\ket{\phi^{(A)}_a}\big),\big(V\ket{\phi^{(B)}_b}\big)\big)$ is a standard pair. We give an example of the corresponding extension of Table \ref{table:ex1} in Table \ref{table:ex2}.

In general, we cannot assume that a measurement table consists of rank-$r$ operators with $r\leq1$. For example, it is not obvious whether the perfectly distinguishable pair given in Table~\ref{table:ex3} can be embedded in a larger Hilbert space as a standard pair. However, in Section \ref{sec:main}, we show that every perfectly distinguishable pair can be embedded as a standard pair.

\begin{table}
\begin{center}
\renewcommand
\arraystretch{1.1}
\begin{tabular}{cp{30mm}p{30mm}p{20mm}}\hline\hline
 & $\ket{+0}$ & $\ket{+1}$  \\ \hline
$\ket{0+}$ & $[00]$ & $[01]$  \\ 
$\ket{1+}$ & $[10]$ & $[11]$  \\ 
\hline\hline
\end{tabular}
\vspace{0.5cm}
\caption{Corresponding standard pair $\big(\big(V\ket{\phi^{(A)}_a}\big),\big(V\ket{\phi^{(B)}_b}\big)\big)$ of $\big(\big(\ket{\phi^{(A)}_a}\big),\big(\ket{\phi^{(B)}_b}\big)\big)$ defined in Table \ref{table:ex1}, where $V=\sum_{a,b}\ket{ab}\bra{\tilde{\psi}_{ab}}$, $\ket{\tilde{\psi}_{00}}=\frac{1}{2}(\ket{0}+\ket{1}+\ket{2})$, $\ket{\tilde{\psi}_{01}}=\frac{1}{2}(\ket{0}+\ket{1}-\ket{2})$, $\ket{\tilde{\psi}_{10}}=\frac{1}{2}(\ket{0}-\ket{1}+\ket{2})$ and $\ket{\tilde{\psi}_{11}}=\frac{1}{2}(\ket{0}-\ket{1}-\ket{2})$. A measurement table distinguishing the standard pair is also shown in the table.
}
\label{table:ex2} 
\end{center}
\end{table}

\begin{table}
\begin{center}
\renewcommand
\arraystretch{1.1}
\begin{tabular}{cp{30mm}p{30mm}p{20mm}}\hline\hline
 & $\ket{0+3}$ & $\ket{2+4}$  \\ \hline
$\ket{1+2}$ & $[0]+[1]$ & $[2]$  \\ 
$\ket{3+4}$ & $[3]$ & $[4]$  \\ 
\hline\hline
\end{tabular}
\vspace{0.5cm}
\caption{Measurement table to distinguish $\mathbb{S}^{(A)}=(\ket{1+2},\ket{3+4})$ and $\mathbb{S}^{(B)}=(\ket{0+3},\ket{2+4})$.
}
\label{table:ex3} 
\end{center}
\end{table}

\section{Necessary conditions}
We show two propositions regarding necessary conditions for the perfect distinguishability with the help of the post-processing. Since $(\mathbb{S}^{(A)},\mathbb{S}^{(B)})$ given in Table~\ref{table:ex1} saturates both conditions, it can be considered as a maximally non-orthogonal pair in the smallest Hilbert space.

\begin{proposition}
\label{prop:dim}
 If $(\mathbb{S}^{(A)},\mathbb{S}^{(B)})$ is perfectly distinguishable by the post-processing and any pair of a state in $\mathbb{S}^{(A)}$ and a state in $\mathbb{S}^{(B)}$ is non-orthogonal, the dimension of $\mathcal{H}$ must satisfy $\dim\mathcal{H}\geq|\mathbb{S}^{(A)}|+|\mathbb{S}^{(B)}|-1$.
\end{proposition}
\begin{proof}
If either $|\mathbb{S}^{(A)}|$ or $|\mathbb{S}^{(B)}|$ is $1$, the statement is trivial. Thus, we assume $|\mathbb{S}^{(A)}|\geq2$ and $|\mathbb{S}^{(B)}|\geq2$.

 It is enough to show that for any perfectly distinguishable $(\mathbb{S}^{(A)},\mathbb{S}^{(B)})$, the following two conditions cannot be satisfied simultaneously:
 \begin{enumerate}
 \item $\forall a\in\mathbb{K}^{(A)},\forall c\in\{0,1\}, \braket{\phi^{(A)}_a}{\phi^{(B)}_{c}}\neq0$,
 \item $\forall c\in\{0,1\},\ket{\phi^{(B)}_{c}}\in{\rm span}(\mathbb{S}^{(A)}\cup\mathbb{S}^{(B)c})$, where $\mathbb{S}^{(B)c}=\mathbb{S}^{(B)}\setminus (\ket{\phi^{(B)}_{0}},\ket{\phi^{(B)}_{1}})$.
\end{enumerate}
If $(\mathbb{S}^{(A)},\mathbb{S}^{(B)})$ is perfectly distinguishable, we can find a measurement table $\big(M_{ab}\big)$. If the second condition is satisfied, we can find the following decompositions:
\begin{equation}
 \ket{\phi^{(B)}_c}=\sum_{a\in\mathbb{K}^{(A)}}\alpha_{ac}\ket{\phi^{(A)}_a}+\sum_{b\geq2}\beta_{bc}\ket{\phi^{(B)}_{b}}
\end{equation}
for $c\in\{0,1\}$. Since Eq.~\eqref{eq:mtablecond2} implies $M_{a,1-c}\ket{\phi^{(B)}_c}=0$, we obtain 
\begin{equation}
 \forall a\in\mathbb{K}^{(A)},\forall c\in\{0,1\},\alpha_{ac}M_{a,1-c}\ket{\phi^{(A)}_a}=0.
\end{equation}

If the first condition is satisfied, since Eq.~\eqref{eq:mtablecond2} guarantees $\bra{\phi^{(A)}_a}M_{ac}\ket{\phi^{(B)}_c}=\braket{\phi^{(A)}_a}{\phi^{(B)}_c}\neq0$, we obtain 
\begin{equation}
 \forall a\in\mathbb{K}^{(A)},\forall c\in\{0,1\},\alpha_{ac}M_{ac}\ket{\phi^{(A)}_a}\neq0,
\end{equation}
which leads us to a contradiction.
\end{proof}

This proposition shows that the retrieval of the perfect distinguishability of such non-orthogonal pure states appears only with $d(\geq 3)$ dimensional Hilbert space.

\begin{proposition}
If $(\mathbb{S}^{(A)},\mathbb{S}^{(B)})$ is perfectly distinguishable by the post-processing, then $\min\{|\braket{\phi^{(A)}_a}{\phi^{(B)}_b}|^2\}_{(a,b)\in\mathbb{K}}\leq\frac{1}{|\mathbb{S}^{(A)}||\mathbb{S}^{(B)}|}$.
\end{proposition}
\begin{proof}
Let $\big(M_{ab}\big)$ be a measurement table distinguishing $(\mathbb{S}^{(A)},\mathbb{S}^{(B)})$. By using Cauchy-Schwartz inequality, AM-GM inequality, Eq.~\eqref{eq:mtablecond1}, and Eq.~\eqref{eq:mtablecond2}, we can derive the following inequality:
\begin{eqnarray}
 \prod_{ab}|\braket{\phi^{(A)}_a}{\phi^{(B)}_b}|^2&=& \prod_{ab}|\bra{\phi^{(A)}_a}M_{ab}\ket{\phi^{(B)}_b}|^2\nonumber\\
 &\leq&\prod_{ab}\bra{\phi^{(A)}_a}M_{ab}\ket{\phi^{(A)}_a}\bra{\phi^{(B)}_b}M_{ab}\ket{\phi^{(B)}_b}\nonumber\\
 &\leq&\left(\sum_{ab}\frac{\bra{\phi^{(A)}_a}M_{ab}\ket{\phi^{(A)}_a}}{|\mathbb{S}^{(A)}||\mathbb{S}^{(B)}|}\right)^{|\mathbb{S}^{(A)}||\mathbb{S}^{(B)}|}\cdot\nonumber\\
 &&\left(\sum_{ab}\frac{\bra{\phi^{(B)}_b}M_{ab}\ket{\phi^{(B)}_b}}{|\mathbb{S}^{(A)}||\mathbb{S}^{(B)}|}\right)^{|\mathbb{S}^{(A)}||\mathbb{S}^{(B)}|}\nonumber\\
 &=&|\mathbb{S}^{(A)}||\mathbb{S}^{(B)}|^{-|\mathbb{S}^{(A)}||\mathbb{S}^{(B)}|}.
\end{eqnarray}
This completes the proof.
\end{proof}
This proposition shows that there does not exists perfectly distinguishable pair each of whose pair-wise overlap $|\braket{\phi^{(A)}_a}{\phi^{(B)}_b}|$ is strictly larger than the pair given in Table \ref{table:ex1}.

Note that we did not assume that the perfectly distinguishable pair can be embedded as a standard pair in the proofs. This allows us to apply these propositions to a more general setting as discussed in Section \ref{sec:MKP}.

\section{Perfect distinguishability as a matrix decomposition}
\label{sec:main}
We show that the perfect distinguishability with the help of the post-processing can be restated as a matrix-decomposition problem, and give the analytical solution for the problem in the case of $|\mathbb{S}^{(A)}|=|\mathbb{S}^{(B)}|=2$. This result also implies that any perfectly distinguishable pair with the help of the post-processing can be embedded in a larger Hilbert space as a standard pair (see Table \ref{table:ex4}). The main theorem uses Lemma \ref{lemma:decomposition} followed by several definitions about linear algebra.

\begin{table}
\begin{center}
\renewcommand
\arraystretch{1.1}
\begin{tabular}{cp{20mm}p{30mm}p{20mm}}\hline\hline
 & $\ket{10}$ & $\left(\sqrt{\frac{2}{3}}\ket{0}+\sqrt{\frac{1}{3}}\ket{1}\right)\ket{1}$  \\ \hline
$\ket{0}\left(\sqrt{\frac{5}{8}}\ket{0}+\sqrt{\frac{3}{8}}\ket{1}\right)$ & $[00]$ & $[01]$  \\ 
$\ket{1}\left(\sqrt{\frac{1}{4}}\ket{0}+\sqrt{\frac{3}{4}}\ket{1}\right)$ & $[10]$ & $[11]$  \\ 
\hline\hline
\end{tabular}
\vspace{0.5cm}
\caption{Corresponding standard pair $\big(\big(V\ket{\phi^{(A)}_a}\big),\big(V\ket{\phi^{(B)}_b}\big)\big)$ of $\big(\big(\ket{\phi^{(A)}_a}\big),\big(\ket{\phi^{(B)}_b}\big)\big)$ defined in Table \ref{table:ex3}, where $V=\sum_{ab}\ket{ab}\bra{\psi_{ab}}$, where $\ket{\psi_{00}}=\frac{1}{2\sqrt{5}}(-\ket{0}+4\ket{1}+\ket{2}+\ket{3}-\ket{4})$, $\ket{\psi_{01}}=\frac{1}{2\sqrt{3}}(\ket{0}+3\ket{2}-\ket{3}+\ket{4})$, $\ket{\psi_{10}}=\frac{1}{\sqrt{2}}(\ket{0}+\ket{3})$ and $\ket{\psi_{11}}=\frac{1}{\sqrt{6}}(-\ket{0}+\ket{3}+2\ket{4})$.
}
\label{table:ex4} 
\end{center}
\end{table}

\begin{definition}
For two $n$ by $m$ matrices $A$ and $B$, when $A$ is not element-wise smaller than $B$, i.e., $\forall i\in\{1,\dots,n\},\forall j\in\{1,\dots,m\},A_{ij}\geq B_{ij}$, we denote $A\geq B$.
\end{definition}

\begin{definition}
For two $n$ by $m$ matrices $A$ and $B$, when $A$ is element-wise larger than $B$, i.e., $\forall i\in\{1,\dots,n\},\forall j\in\{1,\dots,m\},A_{ij}> B_{ij}$, we denote $A> B$.
\end{definition}

\begin{definition}
 $n$ by $m$ matrices $A$ is called a right stochastic matrix if $A\geq 0$ and $\forall i\in\{1,\dots,n\}, \sum_j A_{ij}=1$, where $x\in\mathbb{R}$ in the matrix (in)equality represents the appropriately sized matrix all of whose element are $x$.
\end{definition}

\begin{definition}
 $n$ by $m$ matrices $B$ is called a left stochastic matrix if $B\geq 0$ and $\forall j\in\{1,\dots,m\}, \sum_i B_{ij}=1$.
\end{definition}

\begin{definition}
 The set of matrices that can be decomposed into the element-wise product of a right stochastic matrix and left one is defined by
 \begin{equation}
 \bar{\mathcal{D}}(n,m):=\{P\in L(\mathbb{R}^m,\mathbb{R}^n)|P=A\circ B\},
\end{equation}
where $\circ$ represents the element-wise product, $L(\mathbb{R}^m,\mathbb{R}^n)$ represents the set of $n$ by $m$ matrices and $A$ and $B$ are a right stochastic matrix and a left one, respectively.
\end{definition}

\begin{definition}
 The set of element-wise positive matrices in $\bar{\mathcal{D}}(n,m)$ is defined by
 \begin{equation}
 \mathcal{D}(n,m):=\{P\in \bar{\mathcal{D}}(n,m)|P>0\}.
\end{equation}
\end{definition}

Note that $\bar{\mathcal{D}}(n,m)$ is the closure of $\mathcal{D}(n,m)$ as shown in Appendix \ref{appendix:closure}.

\begin{lemma}
\label{lemma:decomposition}
 If $n\geq2$ and $m\geq2$, the following statement holds: for any $P\in\bar{\mathcal{D}}(n,m)$ and for any $Q\in L(\mathbb{R}^m,\mathbb{R}^n)$,
\begin{equation}
\label{eq:statement1}
 0\leq Q\leq P\Rightarrow Q\in\bar{\mathcal{D}}(n,m).
\end{equation}
\end{lemma}

\begin{proof}
First, we show that it is sufficient to prove
\begin{equation}
\label{eq:statement2}
  \forall P\in\mathcal{D}(n,m),\forall Q, 0< Q\leq P\Rightarrow Q\in\mathcal{D}(n,m).
\end{equation}
Assume Eq.~\eqref{eq:statement2} holds. Since $\bar{\mathcal{D}}(n,m)$ is the closure of $\mathcal{D}(n,m)$, for any $P\in\bar{\mathcal{D}}(n,m)$ and for any $\delta>0$, there exists $P'\in\mathcal{D}(n,m)$ such that $|P-P'|<\delta$. For any $Q\in L(\mathbb{R}^m,\mathbb{R}^n)$ such that $0\leq Q\leq P$, we define $Q'\in L(\mathbb{R}^m,\mathbb{R}^n)$ as
\begin{equation}
 Q'_{ij}=
 \left\{ \begin{array}{ll}
    Q_{ij} & (0<Q_{ij}\leq P'_{ij}) \\
    P'_{ij} & (Q_{ij}> P'_{ij})\\
    \min\{\delta,P'_{ij}\}&(Q_{ij}=0).
  \end{array} \right.
\end{equation}
Since $0< Q'\leq P'$, $Q'\in\mathcal{D}(n,m)$ by using Eq.~\eqref{eq:statement2}. Note that for any $\epsilon>0$, there exists sufficiently small $\delta>0$ such that $|Q-Q'|<\epsilon$. Thus, $Q\in\bar{\mathcal{D}}(n,m)$.

Second, we show that it is sufficient to prove
\begin{equation}
\label{eq:statement3}
  \forall P\in\mathcal{D}(2,2),\forall Q, 0< Q\leq P\Rightarrow Q\in\mathcal{D}(2,2).
\end{equation}
Note that for proving Eq.~\eqref{eq:statement2}, it is sufficient to prove for any $i\in\{1,\dots,n\}$ and $j\in\{1,\dots,m\}$ and for any $\delta\in(0,1]$,
\begin{equation}
 \label{eq:statement4}
  \forall P\in\mathcal{D}(n,m), P\circ T^{(ij)}_{\delta}\in\mathcal{D}(n,m),
\end{equation}
where $T^{(ij)}_\delta$ is the $n$ by $m$ matrix all of whose elements are $1$ except the $(i,j)$ element, which is set to $\delta$. Assume Eq.~\eqref{eq:statement3} holds. For any $P\in\mathcal{D}(n,m)$ and for any $T^{(ij)}_\delta$, pick up their arbitrary $2$ by $2$ submatrices $P[2]$ and $T^{(ij)}_\delta[2]$ containing the $(i,j)$ element. (There exist such submatrices since we assume $n\geq2$ and $m\geq2$.) Letting $P=A\circ B$, the corresponding submatrices $A[2]$ and $B[2]$ satisfy $P[2]=A[2]\circ B[2]$. Define right stochastic matrix $\tilde{A}[2]$ and left one $\tilde{B}[2]$ by
\begin{eqnarray}
\tilde{A}[2]&:=&A[2]\circ  \begin{pmatrix}
 \frac{1}{A[2]_{1*}}& \frac{1}{A[2]_{1*}}\\
  \frac{1}{A[2]_{2*}}& \frac{1}{A[2]_{2*}}
\end{pmatrix}\\
\tilde{B}[2]&:=&B[2]\circ  \begin{pmatrix}
 \frac{1}{B[2]_{*1}}& \frac{1}{B[2]_{*2}}\\
  \frac{1}{B[2]_{*1}}& \frac{1}{B[2]_{*2}}
\end{pmatrix},
\end{eqnarray}
where $A[2]_{i*}=A[2]_{i1}+A[2]_{i2}$ and $B[2]_{*j}=B[2]_{1j}+B[2]_{2j}$. Since $0<\tilde{A}[2]\circ\tilde{B}[2]\circ T^{(ij)}_{\delta}[2]\leq \tilde{A}[2]\circ\tilde{B}[2]$, there exists right stochastic matrix $\tilde{A}'[2]$ and left one $\tilde{B}'[2]$ satisfying $\tilde{A}'[2]\circ\tilde{B}'[2]=\tilde{A}[2]\circ\tilde{B}[2]\circ T^{(ij)}_{\delta}[2]$ by using Eq.~\eqref{eq:statement3}. Define element-wise positive $2$ by $2$ matrices $A'[2]$ and $B'[2]$ by
\begin{eqnarray}
A'[2]&:=&\tilde{A}'[2]\circ  \begin{pmatrix}
 A[2]_{1*}& A[2]_{1*}\\
  A[2]_{2*}& A[2]_{2*}
\end{pmatrix}\\
B'[2]&:=&\tilde{B}'[2]\circ  \begin{pmatrix}
 B[2]_{*1}& B[2]_{*2}\\
  B[2]_{*1}& B[2]_{*2}
\end{pmatrix}.
\end{eqnarray}
Since $A'[2]\circ B'[2]=P[2]\circ T^{(ij)}_\delta[2]$ and $A$ $(B)$ whose submatrix $A[2]$ $(B[2])$ is replaced by $A'[2]$ $(B'[2])$ is also a right (left) stochastic matrix, Eq.~\eqref{eq:statement4} is proven.

Third, we prove Eq.~\eqref{eq:statement3} by explicitly analyzing $\mathcal{D}(2,2)$. By definition, $P\in\mathcal{D}(2,2)$ if and only if $P>0$ and there exist real numbers $A_{21}$, $A_{22}$, $B_{11}$, $B_{12}$ and $A_{11}\in(P_{11},1-P_{12})$ such that
\begin{equation}
\label{eq:decomposition1}
\begin{pmatrix}
 P_{11}&P_{12}\\
 P_{21}&P_{22}
\end{pmatrix}=
\begin{pmatrix}
 A_{11}&1-A_{11}\\
 A_{21}&A_{22}
\end{pmatrix}\circ
\begin{pmatrix}
 B_{11}&B_{12}\\
 1-B_{11}&1-B_{12}
\end{pmatrix}
\end{equation}
and $A_{21}+A_{22}=1$. Note that two conditions $A_{11}\in(P_{11},1-P_{12})$ and $P>0$ are necessary and sufficient for two matrices on the right hand side of Eq.~\eqref{eq:decomposition1} to be element-wise positive. Under the two conditions, $A_{21}+A_{22}$ can be regarded as a function of $A_{11}$ defined by 
\begin{equation}
 f(A_{11})=\frac{P_{21}}{1-\frac{P_{11}}{A_{11}}}+\frac{P_{22}}{1-\frac{P_{12}}{1-A_{11}}}.
\end{equation}
Thus, $P\in\mathcal{D}(2,2)$ if and only if $P>0$ and there exists real number $x\in(P_{11},1-P_{12})$ such that $f(x)=1$. If $P_{11}<1-P_{12}$, $f$ is an unbounded convex function ($\lim_{x\searrow P_{11}}f(x)= \lim_{x\nearrow 1-P_{12}}f(x)=\infty$) with global minimum $f(x^*)$, where $x^{*}=\lambda P_{11}+(1-\lambda)(1-P_{12})$ and $\lambda=\frac{\sqrt{P_{12}P_{22}}}{\sqrt{P_{11}P_{21}}+\sqrt{P_{12}P_{22}}}$. By straightforward calculation, $P\in\mathcal{D}(2,2)$ if and only if 
\begin{eqnarray}
 (P>0)\wedge (P_{11}+P_{12}<1)\ \ \ \ \ \ \ \ \ \ \ \ \ \ \ \ \ \ \nonumber\\
 \label{eq:D2}
 \wedge(P^c_{11} P^c_{22}+ P^c_{12} P^c_{21}-2(P_{11}P_{12}P_{21}P_{22})^{\frac{1}{2}}\geq 1),\ \ \ \ \ 
\end{eqnarray}
where $P^c_{ij}=1-P_{ij}$. This implies Eq.~\eqref{eq:statement3}.
\end{proof}

\begin{theorem}
\label{thm:main}
Assume $|\mathbb{S}^{(A)}|\geq2$ and $|\mathbb{S}^{(B)}|\geq2$. The following three conditions are equivalent:
\begin{enumerate}
 \item $(\mathbb{S}^{(A)},\mathbb{S}^{(B)})$ is perfectly distinguishable by the post-processing
 \item standard pair $\big(\big(\ket{\Phi^{(A)}_a}\big),\big(\ket{\Phi^{(B)}_b}\big)\big)$ exists such that $\braket{\phi^{(A)}_a}{\phi^{(B)}_b}=\braket{\Phi^{(A)}_a}{\Phi^{(B)}_b}$ for all $(a,b)\in\mathbb{K}$
 \item $P\in\bar{\mathcal{D}}(|\mathbb{S}^{(A)}|,|\mathbb{S}^{(B)}|)$, where $P_{ab}=|\braket{\phi^{(A)}_a}{\phi^{(B)}_b}|^2$.
\end{enumerate}
\end{theorem}
\begin{proof}
"$2\Rightarrow 1$" is shown by using Lemma \ref{lemma:isometry} in Appendix \ref{appendix:isometry}, and Proposition \ref{prop:distinguishableset1}. "$3\Rightarrow 2$" is shown by taking the standard pair with the following amplitudes:
\begin{equation}
 \alpha_{ab}=e^{-i\theta(a,b)}\sqrt{A_{ab}},\ \ \ 
 \beta_{ab}=\sqrt{B_{ab}},
\end{equation}
where $e^{i\theta(a,b)}|\braket{\phi_{a}^{(A)}}{\phi_{b}^{(B)}}|=\braket{\phi_{a}^{(A)}}{\phi_{b}^{(B)}}$ and $P=A\circ B$.

We show "$1\Rightarrow 3$" in the following.
If $(\mathbb{S}^{(A)},\mathbb{S}^{(B)})$ is perfectly distinguishable, there exists measurement table $\big(M_{ab}\big)$. Eq.~\eqref{eq:mtablecond1} guarantees that $A_{ab}=\bra{\phi_{a}^{(A)}}M_{ab}\ket{\phi_{a}^{(A)}}$ and $B_{ab}=\bra{\phi_{b}^{(B)}}M_{ab}\ket{\phi_{b}^{(B)}}$ are a right stochastic matrix and left one, respectively. Using Eq.~\eqref{eq:mtablecond2} and Cauchy-Schwartz inequality, we obtain
\begin{equation}
|\braket{\phi_{a}^{(A)}}{\phi_{b}^{(B)}}|^2=\left|\bra{\phi_{a}^{(A)}}M_{ab}\ket{\phi_{b}^{(B)}}\right|^2\leq A_{ab}B_{ab},
\end{equation}
which implies condition $3$ by using Lemma \ref{lemma:decomposition}.
\end{proof}

We can derive the following criteria for the perfect distinguishability as a corollary of Theorem \ref{thm:main} (see Fig.~\ref{fig:region}).

\begin{corollary}
 Assume $|\mathbb{S}^{(A)}|=|\mathbb{S}^{(B)}|=2$. Let $2$ by $2$ matrix $P$ be $P_{ab}=|\braket{\phi^{(A)}_a}{\phi^{(B)}_b}|^2$. Then, $(\mathbb{S}^{(A)},\mathbb{S}^{(B)})$ is perfectly distinguishable by the post-processing if and only if $P$ satisfies
 \begin{equation}
 \label{eq:region}
 P^c_{11} P^c_{22}+ P^c_{12} P^c_{21}-2(P_{11}P_{12}P_{21}P_{22})^{\frac{1}{2}}\geq 1,
\end{equation}
where $P^c_{ij}=1-P_{ij}$.
\end{corollary}
A proof is straightforward by using Eq.~\eqref{eq:D2} and the fact that $\bar{\mathcal{D}}(2,2)$ is the closure of $\mathcal{D}(2,2)$. Note that similar criteria for larger sets can be analytically obtained via a similar derivation of Eq.~\eqref{eq:D2}.

\begin{figure}
 \centering
  \includegraphics[height=.25\textheight]{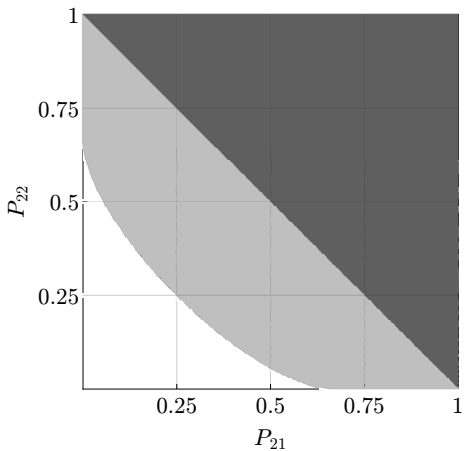}
  \caption{The region of $(P_{21},P_{22})$ for perfectly distinguishable $(\mathbb{S}^{(A)},\mathbb{S}^{(B)})$ with the help of the post-processing when $P_{11}=P_{12}=1/4$, shown by the white region. The example shown in Table \ref{table:ex1} resides on the boundary of perfectly distinguishable pairs. Note that since $\mathbb{S}^{(B)}$ is an indexed set of orthonormal vectors, $(P_{21},P_{22})$ cannot be in the dark gray region for any $(\mathbb{S}^{(A)},\mathbb{S}^{(B)})$.}
\label{fig:region}
\end{figure} 

\section{Related past work}
\label{sec:MKP}
The investigation of perfectly distinguishable tuple $\big(\mathbb{S}^{(n)}=\big(\ket{\phi_{k}^{(n)}}\big)_{k\in\mathbb{K}^{(n)}}\big)_{n=1}^N$ with the help of the post-processing of measurement outcomes with label $n$ is related to the mean king's problem (MKP) \cite{MKP1,MKP2,MKP3,MKP4,MKP5,MKP6}. The MKP consists of three steps: first, a player prepares composite system $\mathcal{H}\otimes\mathcal{R}$. Second, the mean king performs a randomly chosen projective measurement on subsystem $\mathcal{H}$. Third, the player tries to guess the king's measurement outcome by post-processing of her own measurement outcomes obtained by measuring $\mathcal{H}\otimes\mathcal{R}$ and the label of the measurement chosen by the king. The main issue in the MKP---understanding the ensemble of the king's measurement whose outcome can be perfectly identified by the player---has led to the development of several important concepts in quantum mechanics, including mutually unbiased basis \cite{MUB1,MUB2} and a weak value \cite{WV}.

It is known that even for non-commuting projective measurements which inevitably produce non-orthogonal pure states for distinct outcomes in the third step, the player can still identify the king's outcome perfectly with the help of the post-processing. Thus, the retrieval of the perfect distinguishability of non-orthogonal pure states can partially be understood by using the result of the MKP. However, since the king cannot prepare general non-orthogonal pure states in $\mathcal{H}\otimes\mathcal{R}$ by interacting only with the subsystem $\mathcal{H}$, a full understanding of the phenomenon cannot be obtained via the MKP. On the other hand, in many cases, it is enough for the player to prepare the maximally entangled state in the first step of the MKP \cite{MKP1,MKP3,MKP4,MKP5, MKP6,MUB2}. In such cases, the only non-trivial part of the problem is whether the non-orthogonal pure states produced in the third step is perfectly distinguishable with the help of the post-processing. Therefore, the investigation of perfectly distinguishable tuple $\big(\mathbb{S}^{(n)}\big)_{n=1}^N$ with the help of the post-processing extracts an intriguing structure from the MKP as a simpler problem, which would deepen our understanding of the MKP and lead us to key concepts in quantum mechanics.

As a first step toward the general case, we have investigated the case of $N=2$. Note that the three propositions we have shown hold for general $N$, which could be a guide to the further investigation for the general case.

\section{Conclusion}
We have investigated perfectly distinguishable pair of ensembles of pure states $(\mathbb{S}^{(A)},\mathbb{S}^{(B)})$ with the help of post-processing, and have shown that such a pair can always be embedded in a larger Hilbert space as a corresponding standard pair. The distinguishability has been shown to be completely determined by whether a matrix whose elements consist of $|\braket{\phi^{(A)}_a}{\phi^{(B)}_b}|^2$ can be decomposed into the element-wise product of two types of stochastic matrices. By using the result, we also gave a complete characterization of perfectly distinguishable pairs when $|\mathbb{S}^{(A)}|=|\mathbb{S}^{(B)}|=2$. Furthermore, we gave necessary conditions for $N$-tuple $\big(\mathbb{S}^{(n)}\big)_{n=1}^N$ to be perfectly distinguishable by the post-processing.

\begin{acknowledgments}
We are greatly indebted to Seiichiro Tani, Yuki Takeuchi, Yasuhiro Takahashi, Takuya Ikuta, Hayata Yamasaki, Akihito Soeda, Mio Murao, Tomoyuki Morimae, Robert Salazar and Teiko Heinosaari for their valuable discussions. 
\end{acknowledgments}

\appendix
\section{Probabilistic post-processing}
\label{appendix:probpostprocessing}
 A general post-processing can be described by conditional probability distributions $p^{(X)}(\hat{k}|\omega)$ with measurement outcome $\omega$, label of the subset $X$, and guess $\hat{k}$. Under this generalization, $(\mathbb{S}^{(A)},\mathbb{S}^{(B)})$ is perfectly distinguishable if and only if there exist POVM $\big(M_\omega\big)_{\omega\in\Omega}$ and generalized post-processing $\big(p^{(X)}(\hat{k}|\omega)\big)_{X\in \{A,B\}}$ such that
\begin{equation}
\label{eq:defofperfect2}
 \forall X,\forall k, \sum_{\omega\in\Omega}p^{(X)}(k|\omega)\bra{\phi^{(X)}_k}M_\omega\ket{\phi^{(X)}_k}=1.
\end{equation}
We show that there exist POVM $\big(M_\omega\big)_{\omega\in\Omega}$ and generalized post-processing $\big(p^{(X)}(\hat{k}|\omega)\big)_{X\in \{A,B\}}$ satisfying Eq.~\eqref{eq:defofperfect2} if and only if there exist POVM $\big(M_\omega\big)_{\omega\in\Omega}$ and post-processing $\big(f^{(X)}\big)_{X\in \{A,B\}}$ satisfying Eq.~\eqref{eq:defofperfect1}. The only non-trivial part is the "only if" part. Assume there exist POVM $\big(M_\omega\big)_{\omega\in\Omega}$ and generalized post-processing $\big(p^{(X)}(\hat{k}|\omega)\big)_{X\in \{A,B\}}$ satisfying Eq.~\eqref{eq:defofperfect2}. Since $\sum_{\omega\in\Omega}\bra{\phi^{(X)}_k}M_\omega\ket{\phi^{(X)}_k}=1$ and $p^{(X)}(k|\omega)\leq1$, $\bra{\phi^{(X)}_k}M_\omega\ket{\phi^{(X)}_k}>0$ implies $p^{(X)}(k|\omega)=1$. Therefore, we can represent $\Omega$ as the union of its disjoint subsets,
\begin{eqnarray}
 \Omega^{(X)}_{\bot}&:=&\{\omega\in\Omega|\forall k\in\mathbb{K}^{(X)},\bra{\phi^{(X)}_k}M_\omega\ket{\phi^{(X)}_k}=0\}\ \ \\
 \Omega^{(X)}_{k}&:=&\{\omega\in\Omega|\bra{\phi^{(X)}_k}M_\omega\ket{\phi^{(X)}_k}>0\}.
\end{eqnarray}
Define $f^{(X)}$ as $f^{(X)}(\omega)=k$ for $\omega\in\Omega_{k}^{(X)}$, and let $f^{(X)}(\omega)$ be an arbitrary value in $\mathbb{K}^{(X)}$ for $\omega\in\Omega^{(X)}_{\bot}$.
Then, we can verify that such $\big(f^{(X)}\big)_{X\in\{A,B\}}$ satisfies Eq.~\eqref{eq:defofperfect1}.

\section{Analytical property of $\mathcal{D}(n,m)$}
\label{appendix:closure}
In this appendix, we show that $\bar{\mathcal{D}}(n,m)$ is the closure of $\mathcal{D}(n,m)$ relative to metric space $L(\mathbb{R}^m,\mathbb{R}^n)$. By definition, $\mathcal{D}(n,m)\subseteq\bar{\mathcal{D}}(n,m)$ and $\bar{\mathcal{D}}(n,m)$ is closed. Take arbitrary element $P\in\bar{\mathcal{D}}(n,m)$ such that $P\notin \mathcal{D}(n,m)$. Let $P=A\circ B$, where $A$ and $B$ are a right stochastic matrix and left one, respectively. Let $\vec{A}_i=(A_{i1},\dots,A_{im})$ and $j_{\max}(i)$ be a function satisfying $\forall i, \vec{A}_i\leq A_{i,j_{\max}(i)}$. Define $n$ by $m$ matrix $A'$ as
\begin{equation}
 A'_{ij}=
  \left\{ \begin{array}{ll}
    0 & (\vec{A}_i>0) \\
    -\frac{A_{ij}}{2} & ({\rm else\ if\ }j=j_{\max}(i))\\
   \frac{A_{i,j_{\max}(i)}}{2(m-1)}&({\rm otherwise}).
  \end{array} \right.
\end{equation}
Then, we can verify that $A+\delta A'$ is an entrywise-positive and right stochastic matrix for any $\delta\in(0,1]$. By defining $n$ by $m$ matrix $B'$ in a similar manner, we can check $Q:=(A+\delta A')\circ(B+\delta B')\in\mathcal{D}(n,m)$, and it satisfies
\begin{eqnarray}
 |P-Q|&=&|(A+\delta A')\circ(B+\delta B')-A\circ B|\\
 { }&\leq& \delta(|A'\circ B|+|A\circ B'|+\delta|A'\circ B'|).
\end{eqnarray}
Thus, for any $\epsilon>0$, there exists $Q\in\mathcal{D}(n,m)$ such that $|P-Q|<\epsilon$, i.e., $P$ is a limit point of $\mathcal{D}(n,m)$. This completes the proof.

\section{Existence of isometry}
\label{appendix:isometry}
We prove the following lemma used in the proof of Theorem \ref{thm:main}.
\begin{lemma}
 \label{lemma:isometry}
 If $\big(\ket{\psi_i}\in\mathcal{H}\big)_{i\in\mathbb{I}}$ and $\big(\ket{\Psi_i}\in\mathcal{H}'\big)_{i\in\mathbb{I}}$ satisfy $\braket{\psi_i}{\psi_j}=\braket{\Psi_i}{\Psi_j}$ for all $i,j\in\mathbb{I}$, there exists isometry $V:\tilde{\mathcal{H}}\rightarrow\mathcal{H}'$ such that $V\ket{\psi_{i}}=\ket{\Psi_i}$ for all $i\in\mathbb{I}$, where $\tilde{\mathcal{H}}={\rm span}\big(\{\ket{\psi_i}\}_{i\in\mathbb{I}}\big)$ and $\mathbb{I}$ is a finite set.
\end{lemma}
\begin{proof}
Take a basis of $\tilde{\mathcal{H}}$ as $\{\ket{\psi_i}\}_{i\in\tilde{\mathbb{I}}}$, where $\tilde{\mathbb{I}}\subseteq\mathbb{I}$. Define linear operator $V:\tilde{\mathcal{H}}\rightarrow\mathcal{H}'$ as $V\ket{\psi_i}=\ket{\Psi_i}$ for all $i\in\tilde{\mathbb{I}}$. We can easily check that $V$ is an isometry since it does not change the inner product of the basis, i.e., $\bra{\psi_i}V^{\dag}V\ket{\psi_j}=\braket{\Psi_i}{\Psi_j}=\braket{\psi_i}{\psi_j}$ for all $i,j\in\tilde{\mathbb{I}}$. 
 
Let an orthonormal basis of $\tilde{\mathcal{H}}$ be $\big(\ket{\tilde{\psi}_i}\big)_{i\in\tilde{\mathbb{I}}}$. We can verify that indexed set of vectors $\big(\ket{\tilde{\Psi}_i}\big)_{i\in\tilde{\mathbb{I}}}$ defined by $\ket{\tilde{\Psi}_i}=\sum_{j\in\tilde{\mathbb{I}}}\alpha_{ij}\ket{\Psi_j}$ is also orthonormal, where $\alpha_{ij}$ satisfies $\ket{\tilde{\psi}_i}=\sum_{j\in\tilde{\mathbb{I}}}\alpha_{ij}\ket{\psi_j}$. Take arbitrary $j\in\mathbb{I}\setminus\tilde{\mathbb{I}}$ and let $\ket{\psi_j}=\sum_{i\in\tilde{\mathbb{I}}}\beta_{i}\ket{\psi_i}$.

Since $\braket{\tilde{\psi}_k}{\psi_i}=\braket{\tilde{\Psi}_k}{\Psi_i}$ for all $k\in\tilde{\mathbb{I}}$ and $i\in\mathbb{I}$,
\begin{equation}
\label{eq:c2}
\forall k\in\tilde{\mathbb{I}},\  \bra{\tilde{\Psi}_k}\Big(\sum_{i\in\tilde{\mathbb{I}}}\beta_{i}\ket{\Psi_i}\Big)=\braket{\tilde{\psi}_k}{\psi_j}=\braket{\tilde{\Psi}_k}{\Psi_j}.
\end{equation}
Since $\braket{\Psi_j}{\Psi_j}=\braket{\psi_j}{\psi_j}$, $\ket{\Psi_j}=\sum_{i\in\tilde{\mathbb{I}}}\beta_{i}\ket{\Psi_i}$, which shows $V\ket{\psi_j}=\ket{\Psi_j}$ for all $j\in\mathbb{I}\setminus\tilde{\mathbb{I}}$.
\end{proof}

\nocite{*}
\bibliographystyle{eptcs}
\bibliography{generic}

\end{document}